\documentclass{ifacconf}
\usepackage{amsmath}
\usepackage{amssymb}
\newtheorem{assumption}{Assumption}
\newtheorem{definition}{Definition}
\newtheorem{problem}{Problem}

\usepackage{pifont}     
\usepackage{booktabs}   
\usepackage{threeparttable}

\usepackage{xcolor}
\usepackage{natbib}        
\definecolor{mycolor}{HTML}{0072B2} 
\newcommand{\colorsquare}[1]{\textcolor{#1}{\rule{1ex}{1ex}}}
\definecolor{mycolor1}{HTML}{D55E00}
\definecolor{mycolor2}{HTML}{A50026}

\newcommand{\cmark}{\ding{51}} 
\newcommand{\xmark}{\ding{55}} 

\newtheorem{theorem}{Theorem}

\newtheorem{remark}{Remark}
\usepackage{graphicx}      

\usepackage{xcolor}

\newif\ifshowSSR
\showSSRtrue   

\usepackage[normalem]{ulem} 
\usepackage{xcolor}

\begin{document}
\begin{frontmatter}

\title{Approximation-Free Control Barrier Functions for Prescribed-Time Reach-Avoid of Unknown Systems\thanksref{footnoteinfo}}

\thanks[footnoteinfo]{The work is supported in part by the ARTPARK}

\author[First]{Shubham Sawarkar} 
\author[First]{Pushpak Jagtap}

\address[First]{Indian Institute of Science (IISc), Bangalore, India. \\(e-mail: \{shubhamsg,pushpak\}@iisc.ac.in)}

\begin{abstract}                
We study the prescribed-time reach-avoid (PT-RA) control problem for nonlinear systems with unknown dynamics operating in environments with moving obstacles. Unlike robust or learning-based Control Barrier Function (CBF) methods, the proposed framework requires neither online model learning nor uncertainty bound estimation. A CBF-based Quadratic Program (CBF-QP) is solved on a simple virtual system to generate a safe reference satisfying PT-RA conditions with respect to time-varying, tightened obstacle and goal sets. The true system is confined to a \emph{Virtual Confinement Zone} (VCZ) around this reference using an approximation-free feedback law. This construction guarantees real-time safety and prescribed-time target reachability under unknown dynamics and dynamic constraints without explicit model identification or offline precomputation. Simulation results illustrate reliable dynamic obstacle avoidance and timely convergence to the target set.
\end{abstract}

\begin{keyword}
Control barrier functions; Reach–avoid problems; Unknown dynamical systems.
\end{keyword}

\end{frontmatter}

\section{Introduction}
Safety assurance in dynamical systems is essential in applications such as autonomous vehicles, aerial robots, and industrial automation. Beyond achieving the control objective, it is necessary to ensure that the system trajectories remain within safe regions of the state space. These requirements are often expressed as state constraints, and most existing control strategies rely on some knowledge of the system dynamics to check or enforce them. In practice, obtaining an accurate model is difficult due to uncertainty, unmodeled effects, and environmental changes. Exhaustive testing is also not feasible, yet safety and timely goal satisfaction must still be ensured, often within a prescribed time window.

Several approaches have been developed for the safety-critical control of nonlinear systems. Hamilton–Jacobi (HJ) reachability~\citep{mitchell2005} provides rigorous guarantees through backward reachable sets, but it requires solving high-dimensional PDEs, making it model-dependent and computationally expensive. Predictive safety filters~\citep{waberisich2023} enforce Nagumo’s condition through constrained optimization, but their complexity grows rapidly with the prediction horizon and the number of known obstacle trajectories. \cite{saxena} proposed an RL-based approach to handle partially unknown dynamics, but it requires extensive offline training and offers no formal safety guarantees. Symbolic control~\citep{tabuada2009verification} requires large offline abstractions, making it impractical for high-dimensional or prescribed-time tasks. Achieving prescribed-time reach–avoid behavior remains challenging for all these approaches.

CBFs~\citep{ames2019theory,jagtap2020formal} have emerged as an efficient framework for real-time safety enforcement. However, classical CBF-based Quadratic Programs (CBF-QPs)~\citep{ames2019theory} assume fully known dynamics. To address uncertainty, several extensions have been proposed. Robust formulations explicitly compensate for known bounded disturbances or modeling errors~\citep{JANKOVIC2018359,buch2022,XU201554,kolathaya2019,alan2023}, while learning-based methods employ Gaussian Processes or neural networks to approximate the unknown dynamics from data~\citep{jagtap2020,seiler2022,COHEN2024100947}. {\cite{molnar2021} proposed a model-free safety-critical control framework where a safe velocity is synthesized via CBFs on a reduced-order kinematic model and tracked by a platform-specific controller. However, the model-free variant provides only input-to-state safety (ISSf) guarantees, and the framework does not address prescribed-time reach or moving obstacles.} Other works focus on improving estimation and observer design under uncertainty~\citep{takano2018} or adaptively adjusting safety margins. Moreover, dynamic obstacle avoidance has been addressed via collision-cone constraints~\citep{manan}.
 Despite these advances, most approaches still depend on either prior knowledge of uncertainty bounds or extensive offline training, limiting their practicality for systems with unstructured or time-varying uncertainties. Time-varying CBFs~\citep{8404080} extend safety constraints to logical or temporal tasks but remain applicable mainly to known system models and known disturbance bounds.

To deal with an unknown model, Prescribed-Performance Control~\citep{4639441,BERGER2018345,Berger2021} provides an approximation-free framework that ensures tracking errors remain in user-defined transient performance bounds. ~\citet{dasstt} and \cite{das2025} extend the approximation-free framework to address the reach-avoid problem by using Spatiotemporal Tubes (STT). However, these frameworks do not explicitly encode safety constraints or support dynamic obstacle avoidance.

To overcome these limitations, we propose the \emph{Virtual Confinement Zone (VCZ)} framework, which unifies approximation-free confinement and CBF-based safety. A CBF-QP generates a virtual trajectory ensuring reach-avoid satisfaction within a prescribed time using time-varying CBF constraints, while the true system is confined around this trajectory inside a shrinking, time-varying region. This construction guarantees VCZ forward invariance and prescribed-time reachability for nonlinear systems with unknown dynamics and moving obstacles, providing real-time safety and goal satisfaction without exact model knowledge, conservative uncertainty bounds, or any form of offline pre-computation. The proposed framework has been demonstrated using simulation.

\section{Problem Formulation}
\label{sec:pblm_formulation}
\textit{Notations:} 

The Euclidean norm of vector $x\in\mathbb{R}^n$ is denoted by \(\|x\|\). 
The closed and open balls centered at \(x_0\) with radius \(r\) are given by 
\(\mathcal{B}(x_0, r) := \{x \in \mathbb{R}^n \mid \|x - x_0\| \le r\}\) and 
\(\mathcal{B}^\circ(x_0, r) := \{x \in \mathbb{R}^n \mid \|x - x_0\| < r\}\), respectively. 
The intersection and union of a collection of \(d\) sets 
\(\{\mathcal{S}_i\}_{i=1}^d\) are denoted by 
\(\bigcap_{i=1}^d \mathcal{S}_i\) and 
\(\bigcup_{i=1}^d \mathcal{S}_i\), respectively.
For a vector-valued function \( h:\mathbb{R}^n \to \mathbb{R}^m \), 
the Jacobian with respect to \(x\) is denoted by 
\( \mathbb{J}_x h := \frac{\partial h(x)}{\partial x} \in \mathbb{R}^{m\times n} \).
The class \(\mathcal{K}\) consists of continuous, strictly increasing functions \(\alpha\) with \(\alpha(0) = 0\); 
\(\mathcal{K}_\infty\) denotes unbounded \(\mathcal{K}\) functions; 
and the extended \(\mathcal{K}_\infty\) class includes functions \(\alpha : \mathbb{R} \to \mathbb{R}_+\) that are continuous, strictly increasing, unbounded, and satisfy \(\alpha(0) = 0\). 
The notations \(A \succ (\prec) 0\) and \(A \succeq (\preceq) 0\) represent positive (negative) definite and semidefinite matrices, respectively.
The symbol \(0_n\) denotes a vector of \(n\) zeros. All other notation in this paper follows standard mathematical conventions.
\subsection{System Definition}
We consider a nonlinear control-affine dynamical system 
\begin{equation}
    \dot{x} = f(x) + g(x)u + \omega,
    \label{eqn:sys_dyn}
\end{equation}
where \( x(t) \in \mathbb{R}^n \) denotes the state, \( u(t
)\in \mathbb{R}^n \) denotes the control input, maps \( f : \mathbb{R}^n \to \mathbb{R}^n \), \( g : \mathbb{R}^n \to \mathbb{R}^{n \times n} \) and disturbance signal \( \omega : \mathbb{R}_+ \to \mathbb{R}^n \) satisfies the following assumptions. 
\begin{assumption}
    The functions $f$ and $g$ are unknown, but bounded and locally Lipschitz continuous. Moreover, $\frac{g(x)+g(x)^\top}{2}$ is assumed to be sign definite with known sign for all $x\in\mathbb{R}^n$. The disturbance $\omega$ is an unknown but bounded and piecewise continuous signal.
 \label{ass: assm1}
\end{assumption}
\subsection{Control Barrier Functions (CBFs)}
A Control Barrier Function (CBF), as defined in \cite{ames2019theory}, is a scalar function that ensures forward invariance of a safe set by enforcing an inequality constraint on the system’s input.

\begin{definition}
\label{def:CBFQP}

A continuously differentiable vector-valued function 
\( h : \mathbb{R}^n \times \mathbb{R}_{+} \to \mathbb{R}^d \), where \(d\) is the number of 
unsafe sets, is given by
\(
h(z,t) := [h_1(z,t), h_2(z,t), \dots, h_d(z,t)]^\top,
\)
where each scalar function \(h_i\) is a candidate \emph{Control Barrier Function (CBF)} defining 
the unsafe set
\begin{equation}
\mathcal{U}_i(t) := \{ z \in \mathbb{R}^n \mid h_i(z,t) < 0\}, \quad \forall i \in \{1, \dots, d \},
\end{equation}
with \( h_i(z,t) < 0 \) in the interior of \(\mathcal{U}_i\), \( h_i(z,t) = 0 \) on 
its boundary, and \( h_i(z,t) \geq 0 \) outside.
\end{definition}

Given a known control-affine system \( \dot{z} = f_z(z) + g_z(z)u_z \) where \(f_z:\mathbb{R}^n \rightarrow \mathbb{R}^n,g_z:\mathbb{R}^{n}\rightarrow\mathbb{R}^{n\times m}\) and \(u_z \in \mathbb{R}^m\), forward invariance over \(\mathbb{R}^n\setminus\mathcal{U}_i(t)\) is ensured by enforcing the following CBF condition for all \( i \in \{1, \dots, d\} \):
\begin{equation}
\label{eq:cbf_condition}
\frac{\partial h_i^\top(z,t)}{\partial z}(f_z(z) + g_z(z)u_z) + \frac{\partial h_i(z,t)}{\partial t} \!\geq\! -\gamma_i(h_i(z,t)),
\end{equation}
where \( \gamma_i : \mathbb{R} \to \mathbb{R} \) is an extended class \( \mathcal{K}_\infty \) function.

The control input \( u_z \in \mathbb{R}^m \) is then synthesized via the following Quadratic Program (QP):
\begin{equation}
\label{eq:qp_general}
\begin{aligned}
&\min_{u_z \in \mathbb{R}^m} \ \frac{1}{2} u_z^\top H u_z + F^\top u_z \\
&\text{ s.t.} \ \mathbb{J}_z h^\top(z,t)(f_z(z)+g_z(z)u_z)\hspace{-0.2em} + \hspace{-0.2em}\frac{\partial h(z,t)}{\partial t} \hspace{-0.2em}+ \hspace{-0.2em}\Gamma(h(z,t)) \geq 0,
\end{aligned}
\end{equation}
where $\Gamma(h(z,t)) := [\gamma_1(h_1(z,t)),\cdots, \gamma_d(h_d(z,t))]^\top$,
\(H \in \mathbb{R}^{m \times m}\) is a positive-definite Hessian matrix and \(F\in \mathbb{R}^m\) is a gradient vector.
This QP guarantees that the system remains within the safe set \(\mathbb{R}^n\setminus\bigcup_{i=1}^d \mathcal{U}_i\).

\begin{lem}[Forward Invariance of the Safe Set]
\label{lem:cbfqp}
Let \( h_i : \mathbb{R}^n \times \mathbb{R}_+ \rightarrow \mathbb{R} \), \( i \in \{1, \dots, d\} \), be continuously differentiable functions defining the unsafe set \(\bigcup_{i=1}^d \mathcal{U}_i \)
If the control input \( u_z \) is chosen as the solution to the Quadratic Program \eqref{eq:qp_general}, and the QP remains feasible for all \( z \in \mathbb{R}^n\setminus\bigcup_{i=1}^d \mathcal{U}_i(t)  \) and $t\in\mathbb{R}_+$, then the safe set \(\mathbb{R}^n\setminus\bigcup_{i=1}^d \mathcal{U}_i(t)\) is forward invariant. That is,
\[
z(0) \in \mathbb{R}^n\setminus\bigcup_{i=1}^d \mathcal{U}_i(0)  \quad \Rightarrow \quad z(t) \in \mathbb{R}^n\setminus\bigcup_{i=1}^d \mathcal{U}_i(t), \quad \forall t \in \mathbb{R}_+.
\]
\end{lem}

\subsection{Problem Formulation}
In this work, we consider the prescribed-time reach-avoid tasks as defined next.
\begin{definition}
\label{def:ptrs}
Given an initial state \( x(0) \in \mathbb{R}^n \), let \( \mathcal{U}(t) \subset \mathbb{R}^n \) denote the time-varying unsafe set, and let \( \mathcal{R} \subset \mathbb{R}^n \setminus \mathcal{U}(t) \) denote the target set. For the initial state \( x(0) \in \mathbb{R}^n \setminus \mathcal{U}(0) \) and a prescribed time \( t_f > 0 \), the system in \eqref{eqn:sys_dyn} is said to satisfy \emph{prescribed-time reach-avoid (PT-RA) task} if
$x(t_f) \in \mathcal{R}$ and $x(t) \in \mathbb{R}^n \setminus \mathcal{U}(t), \ \forall t \in [0, t_f]$.
\end{definition}
\begin{problem}\label{problem1}
Given an unknown system defined in \eqref{eqn:sys_dyn} satisfying Assumption~\ref{ass: assm1}, time-varying unsafe set $\mathcal{U}(t)$, a target set $\mathcal{R}$, a prescribed time $t_f>0$, and initial state \( x(0) \in \mathbb{R}^n \setminus \mathcal{U}(0) \); the objective is to synthesize a continuous controller \( u:\mathbb{R}^n\times \mathbb{R}_+\rightarrow \mathbb{R}^n\) such that the system state \( x(t) \) satisfies the \emph{prescribed-time reach-avoid task} as defined in Definition~\ref{def:ptrs} using a Control Barrier Function-based Quadratic Program (CBF-QP).
\end{problem}

In principle, CBFs can achieve the PT-RA objective by enforcing set invariance and guaranteeing reach-avoid behavior, provided the CBF constraints are feasible. However, this approach relies on precise system dynamics, making it unsuitable for systems with unknown or uncertain dynamics.

To address this, in the next Section, we introduce a \emph{Virtual Confinement Zone} (VCZ), a designer-specified, time-varying region, endowed with nominal virtual dynamics and a virtual control input, that encloses the system trajectory with appropriate approximation-free control design. We then generate a virtual control input via CBF-QP formulation to ensure satisfaction of the PT-RA task by VCZ. Next, as the true system trajectory remains confined within VCZ, it inherits its safety and reachability properties, ensuring the PT-RA task. Detailed controller design and analysis are presented in Sections \ref{Sec:Controller Design} and \ref{sec:analysis}, respectively. The schematic of the control flow is shown in Figure \ref{fig:floww}.

\begin{figure}[tpb] 
    \centering
    \includegraphics[width=0.9\columnwidth]{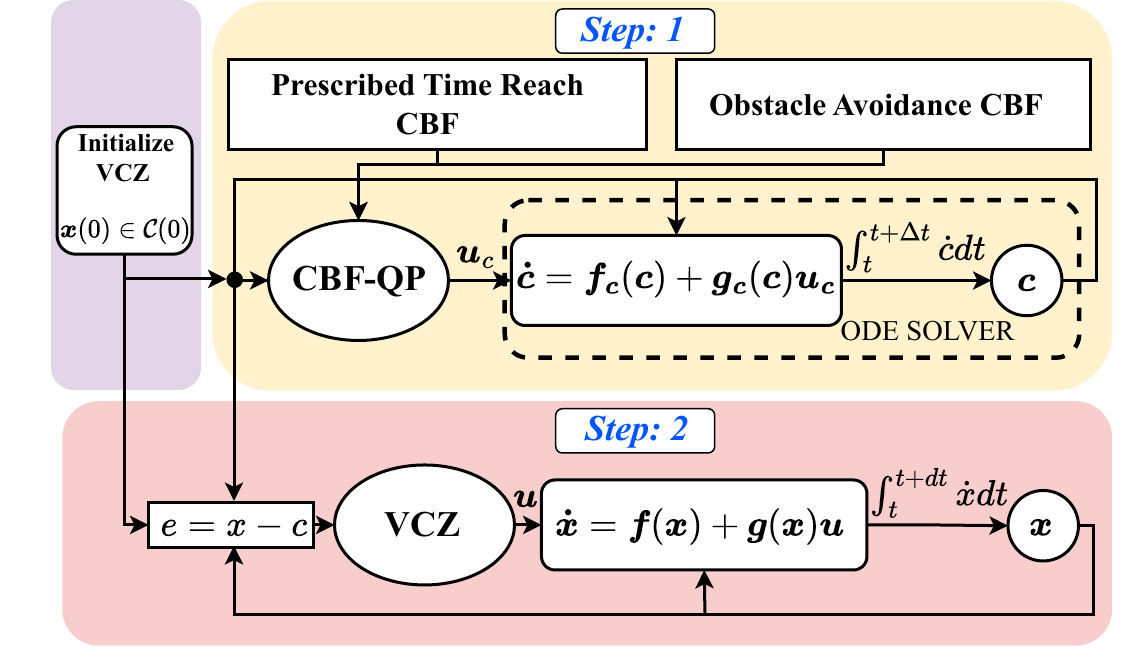}
    \caption{Control Flowchart}
    \label{fig:floww}
\end{figure}

\section{Controller Design} \label{Sec:Controller Design}

To address unknown dynamics, we first introduce the \emph{Virtual Confinement Zone} (VCZ), a user-specified, time-varying region equipped with nominal virtual dynamics and a virtual control input.

\begin{definition}[Virtual Confinement Zone]
\label{def:VCZ}
A \emph{Virtual Confinement Zone} (VCZ) is a time-varying open ball \( \mathcal{C}(t) \subset \mathbb{R}^n \) defined as $
\mathcal{C}(t) := \mathcal{B}^\circ(c(t),r_c)$, where $c(t)$ is a time-varying center and \( r_c > 0 \) is a fixed radius.
\end{definition}

The center \(c(t)\) is assumed to evolve according to a user-defined virtual dynamical system of relative degree one:
\begin{equation}
    \dot{c} = f_c(c) + g_c(c) u_c,
    \label{eqn:virtual_model}
\end{equation}
where \(f_c : \mathbb{R}^n \to \mathbb{R}^n\) and
\(g_c : \mathbb{R}^n \to \mathbb{R}^{n \times m}\) are chosen as locally Lipschitz continuous functions, and 
\(u_c \in \mathbb{R}^m\) is the virtual control input.

The VCZ \(\mathcal{C}(t)\) serves as a dynamically 
adjustable admissible set that guides and constrains the actual system's state. By appropriately designing the virtual dynamics \((f_c, g_c, u_c)\), 
one can generate reference trajectories that respect PT-RA 
requirements. Moreover, the radius \(r_c>0\) is chosen such that the terminal ball
\(\mathcal{B}^\circ(c(t_f), r_c)\) lies entirely within the target region 
\(\mathcal{R}\).

With the (VCZ) and virtual dynamics defined, the overall control strategy proceeds in two concurrent stages:
\begin{itemize}
    \item[(1)] \textit{Virtual control design:} a CBF-QP-based controller that navigates the VCZ center.
    \item[(2)] \textit{Confinement control design:} a separate controller maintains the system state within the evolving VCZ.
\end{itemize}

\subsection*{Step 1: Virtual Control Design using CBF-QP over VCZ dynamics}
\label{subsec: virtual_control}

Consider a prescribed-time reach-avoid problem defined in Problem \ref{problem1}, where 
unsafe region \(\mathcal{U}(t)\) is defined as the union of $d-1$ moving obstacles 
\[
\mathcal{U}(t) :=  \bigcup_{j=1}^{d-1} \mathcal{U}_j(t), \forall t\in[0,t_f],
\]
where each obstacle $\mathcal{U}_j(t)$ is modeled as an open Euclidean ball \(\mathcal{U}_j(t) := \mathcal{B}^\circ(b_{u_j}(t), r_{u_j})\) with center 
\(b_{u_j}(t) \in \mathbb{R}^n\) and radius \(r_{u_j} > 0\).
The target set is a closed ball 
\(\mathcal{R} := \mathcal{B}(b_R, r_R)\) with \(b_R\in \mathbb{R}^n\) and
by Definition~\ref{def:ptrs}, it satisfies
\(\mathcal{R} \subset \mathbb{R}^n\setminus\mathcal{U}(t)\). 

\subsubsection*
{a)~~Initialization of \emph{VCZ}.}
\label{subsec:VCZ_cons}
At \(t=0\), we initialize a Virtual Confinement Zone (VCZ) around the initial state \(x(0)\) [cf. Def \ref{def:VCZ}],  i.e., \(x(0) \in \mathcal{C}(0)\), and 
\(\mathcal{C}(0) \subset \mathbb{R}^n\setminus\mathcal{U}(0)\), ensuring that the system starts in a 
safe configuration. 

Additionally, the radius \(r_c > 0\) of VCZ is chosen as \(r_c < r_R\) to ensure that \(\mathcal{B}^\circ(c(t_f), r_c)\subset \mathcal{R}= \mathcal{B}(b_R, r_R)\).
\subsubsection*{b)~~Barrier functions for PT-RA Task.}
Assuming that the true state lies within \(\mathcal{C}(t)\), an equivalent subobjective is to ensure that the entire VCZ remains safe and reaches the target set \(\mathcal{R}\) at the prescribed time \(t_f\):
\[
\mathcal{C}(t) \subset \mathbb{R}^n \setminus \mathcal{U}(t), \quad \forall t \in [0, t_f]\text{ and }
\mathcal{C}(t_f) \subset \mathcal{R}.
\]

We make the following assumptions about the obstacle sets:  
\begin{assumption}[Obstacle motion and separation]
\label{ass:qp_feasibility}
For all 
\(t\in [0,t_f]\) and for all \(j \in \{1,2,\dots,d-1\}\):
\begin{enumerate}
    \item Each obstacle has known radius \(r_{u_j}\), center position \(b_{u_j}(t)\), 
    and velocity \(\dot{b}_{u_j}(t)\). 
    The quantity  \(\|\dot{b}_{u_j}\|\) is bounded, and obstacle motion \(b_{u_j}(t)\) is continuous.
    \item Minimum distance between any two obstacles satisfies  
    \[
    \min_{i \neq j} \| b_{u_i}(t) - b_{u_j}(t) \| \geq 2r_c+r_{u_i}+r_{u_j},
    \]
    ensuring that the VCZ \(\mathcal{C}(t)\) can pass between them.
    \item At \(t=t_f\), the target set \(\mathcal{R}\) is obstacle-free: \(\mathcal{U}_j(t_f)=\mathcal{B}^\circ(b_{u_j}(t_f), r_{u_j}) \cap \mathcal{R} = \emptyset, \ \forall j\).

\end{enumerate}
\end{assumption}

To ensure the entire VCZ remains outside all obstacles, it suffices for its center $c(t)$ to satisfy the tightened constraint
\[
\|c(t)- b_{u_j}(t)\| \;\ge\; r_{u_j} + r_c, \quad j\in\{1,\dots,d-1\}.
\]
Accordingly, the unsafe set for the \emph{VCZ} center is
\[
\mathcal{U}_{VCZ}(t) := \bigcup_{j=1}^{d-1} 
\mathcal{B}^\circ(b_{u_j}(t), r_{u_j} + r_c).
\]
To guarantee the VCZ remains within the target region \(\mathcal{R}_{VCZ}\) at \(t=t_f\), we define the shrunk target set for the center as 
\[
\mathcal{R}_{VCZ} := \mathcal{B}\big(b_R, r_R - r_c\big),
\]

By construction, maintaining the center of VCZ \(c(t)\) within the tightened safe set \(\mathbb{R}^n\setminus \mathcal{U}_{VCZ}(t)\) for all $t\in[0,t_f]$ and driving it into $\mathcal{R}_{VCZ}$ at $t_f$  ensure that the entire VCZ satisfies PT-RA task.
Formally,
\begin{align}
\label{eqn:reach_avoid}
&\big(c(t)\in \mathbb{R}^n\setminus \mathcal{U}_{VCZ}(t),\ \forall t\in[0,t_f]\big) \wedge\big(c(t_f)\in \mathcal{R}_{VCZ}\big) \notag \\
&\Longrightarrow
\big(\mathcal{C}(t)\subset \mathbb{R}^n\setminus\mathcal{U}(t),\ \forall t\in[0,t_f]\big) \wedge \big(\mathcal{C}(t_f)\subseteq \mathcal{R}\big).
\end{align}
A formal proof of this implication is provided in the next section.

To address the prescribed-time reach-avoid problem for the VCZ center \(c(t)\) with respect to the unsafe set \(\mathcal{U}_{VCZ}(t)\) and the target set \(\mathcal{R}_{VCZ}\), we employ the CBF–QP framework [Def.~\ref{def:CBFQP}] applied to the virtual dynamics \eqref{eqn:virtual_model}. 
We consider \(d-1\) obstacles forming the unsafe set \(\mathcal{U}_{\text{VCZ}}(t)\), with corresponding control barrier functions
\begin{equation}
\label{eqn:avoid_cbf}
    h_j(c, t) := \|c - b_{u_j}(t)\|^2 - (r_{u_j}+r_c)^2,
    \quad j \in\{ 1,2,\dots,d-1\}.
\end{equation}

To ensure $c(t)$ reaches the target \(\mathcal{R}_{VCZ}\), within the prescribed time $t_f$, we construct a shrinking set that initially encloses both \(c(t)\) and \(\mathcal{R}_{VCZ}\) and contracts over time until it is fully contained in 
\(\mathcal{R}_{\text{VCZ}}\) at $t=t_f$. A barrier function associated with this set guarantees the invariance of $c(t)$ and thus the prescribed-time convergence.

Formally, the shrinking set is defined as \(\mathcal{R}_s(t) := \mathcal{B}(b_R, r_r(t)) \subset \mathbb{R}^n\), where \( r_r(t) := -m\frac{t}{t_f} + q \) with initial radius \( r_{r_0} \geq \|c(0) - b_R\| \) at \(t=0\) and final radius\( r_{r_f} \leq r_R - r_c > 0 \)  at \(t=t_f\). The parameters \( m \) and \( q \) are chosen to satisfy the boundary conditions based on the geometry of the sets, and can be explicitly computed as \( m = (r_{r_0} - r_{r_f})/t_f \) and \( q = r_{r_0} \).

 The corresponding barrier function for target reach is 
 \begin{equation}
\label{eqn:reach_cbf}
    h_d(c, t) := r_r(t)^2-\|c - b_R\|^2.
\end{equation}

\subsubsection*{c)~~Controller synthesis for prescribed-time reach-avoid.}
Now we collect \(d-1\) CBFs for obstacle avoidance \eqref{eqn:avoid_cbf} and one CBF for the 
prescribed-time reach \eqref{eqn:reach_cbf} as
\begin{equation}
\label{eqn:cbf_vector}
h(c,t) := [h_1(c,t), h_2(c,t), \dots, h_d(c,t)]^\top.
\end{equation}

By Lemma~\ref{lem:cbfqp}, the CBF conditions guarantee forward invariance of 
the set associated with each CBF 
\begin{equation}
\label{eqn:cbf_confin}
c(t) \in \big(\mathbb{R}^n\setminus \mathcal{U}_{\text{VCZ}}(t)\big) \cap \mathcal{R}_s(t)\subset\mathcal{R}_s(0), \forall t \in [0,\;t_f].    
\end{equation}
Consequently, the VCZ center \(c(t)\) satisfies PT-RA task 
\eqref{eqn:reach_avoid}.

To enforce these conditions, we synthesize the virtual control input \(u_c \in \mathbb{R}^m\) via the following quadratic program (QP) based on the virtual dynamics \eqref{eqn:virtual_model}:
\begin{equation}
\label{eq:cbf_qp_general}
\begin{aligned}
\min_{u_c \in \mathbb{R}^m} \quad & \tfrac{1}{2} u_c^\top H u_c + F^\top u_c \\
\text{s.t.} \quad & \dot{h}(c,t) + \Gamma(h(c,t)) \geq 0,
\end{aligned}
\end{equation}
where \(h(c,t) := [h_1(c,t), \dots, h_d(c,t)]^\top\) and 
\(\Gamma(h(c,t)) := [\gamma_1(h_1(c,t)), \dots, \gamma_d(h_d(c,t))]^\top\). 
The matrices \(H, F\) and the extended class-\(\mathcal{K}_\infty\) function \(\gamma_j(\cdot)\) 
are defined in Def.~\ref{def:CBFQP}.

\begin{assumption}
\label{ass:non_empty_u}
\(\forall t \in [0,t_f]\) and \(\forall c(0)\in \mathbb{R}^n\setminus \mathcal{U}_{VCZ}(0)\), there exists a control input 
\(u_c(t) \in \mathbb{R}^m\) that satisfies the CBF constraints in 
\eqref{eq:cbf_qp_general}.
\end{assumption}
\begin{remark}
From 
\(\dot h_j(c,t) = \tfrac{\partial h_j(c,t)}{\partial c}\big(f_c(c)+g_c(c)u_c\big) 
+ \tfrac{\partial h_j(c,t)}{\partial t}\) in \eqref{eq:cbf_qp_general}, the coefficient of the control input 
is \(\tfrac{\partial h_j(c,t)}{\partial c}\, g_c(c)\). The virtual input matrix 
\(g_c(c)\) must therefore be chosen so that this term does not vanish (at least 
on or near \(h_j(c,t)=0\)) for all \(j\in\{1,\dots,d\}\), ensuring that the CBF 
constraints in \eqref{eq:cbf_qp_general} can be enforced.
\end{remark}
\begin{remark}
 Theoretically, the control framework remains valid for any parameter selection of the virtual dynamics subject to given assumptions and construction. However, in practical implementations, it is advantageous to employ an approximate reference model to enhance performance and facilitate a more effective realization of the control strategy.  
\end{remark}

The virtual dynamics, driven by the virtual control input \(u_c \in \mathbb{R}^m\), ensure the PT-RA task. The next step ensures that the true state remains within the VCZ
\(\mathcal{C}(t)\), which allows the true system to inherit the reach–avoid property.
\subsection*{Step 2: Confinement Control Design}

In Step~1, it has been shown that the entire VCZ \(\mathcal{C}(t)\) avoids the 
true unsafe sets for all $t\in[0,t_f]$ and reaches the target set at time \(t_f\). Hence, if the 
true state trajectory $x(t)$ is confined within the VCZ $\mathcal{C}(t)$, the PT-RA objective in 
Definition~\ref{def:ptrs} is satisfied.
The goal of Step 2 is to ensure this confinement:
\begin{equation}
\label{eq:conf_objective}
x(t) \in \mathcal{C}(t), \qquad \forall t \in [0,t_f].
\end{equation}
To enforce \eqref{eq:conf_objective}, we design a feedback controller \(u\) based on the error between the true state and the VCZ center
\begin{equation}
    e := x - c
\label{eqn:error_def}
\end{equation}
and its normalized form
$\hat{e} := \frac{\|e\|}{r_c}$.
From Definition  \ref{def:VCZ} \( \| x - c \| < r_c \), the normalized error \( \hat{e} \) satisfies \(0 \leq \hat{e} < 1\).

The control input ensuring confinement within VCZ \( \mathcal{C}(t)\;\; \forall t\in [0,t_f] \) is defined as:  
\begin{equation}
u := 
\begin{cases} 
-k\zeta(\hat{e}) \phi(e) & \text{if } \|x-c\| > 0, \\
0_n & \text{if } \|x-c\| = 0,
\end{cases}
\label{eqn:control}
\end{equation}
where \(\zeta : [0,1) \to \mathbb{R}\) is defined as 
\(\zeta(\hat{e}) := \ln\!\left(\tfrac{1+\hat{e}}{1-\hat{e}}\right)\), 
acting as a logarithmic barrier that ensures \(0 \leq \hat{e} < 1\). 
The design gain \(k>0\) if \(\tfrac{g(x)+g(x)^\top}{2} \succ 0\), and \(k<0\) if 
\(\tfrac{g(x)+g(x)^\top}{2} \prec 0\). Furthermore, $\phi(e) := \frac{e}{\|e\|}$.
\begin{remark} {The control input \eqref{eqn:control} decouples magnitude \( k\zeta(\hat{e}) \) 
from direction \( \phi(e) \), which facilitates the stability analysis. 
Lemma~\ref{lem:continuty_of_u} establishes continuity of the controller 
with respect to \( e = x - c \). This structure further lends itself to 
the analysis of bounded control, though such an extension is left for future work.}

For detailed proof, refer to Lemma~\ref{lem:continuty_of_u}
\end{remark}

By Step~1, the VCZ center trajectory \(c(t)\) satisfies the PT-RA
specification with respect to inflated sets \eqref{eqn:reach_avoid}. By Step~2, 
the true state $x(t)$ is confined within the VCZ \(\mathcal{C}(t)\) for all 
\(t\in[0,t_f]\) \eqref{eq:conf_objective}. Together, these yield the implication chain:
\begin{align}\label{eq:chain_conclusion}
&\big(c(t)\in \mathbb{R}^n\setminus \mathcal{U}_{\text{VCZ}}(t),\ \forall t\in[0,t_f]\big)\ 
   \wedge\ \big(c(t_f)\in \mathcal{R}_{\text{VCZ}}\big) \notag\\
&\Longrightarrow\!
   \big(\mathcal{C}(t)\subset \mathbb{R}^n\setminus\mathcal{U}(t),\ \forall t\in[0,t_f]\big)
   \wedge\big(\mathcal{C}(t_f)\subseteq \mathcal{R}\big) \notag\\
&\Longrightarrow\!
   \big(x(t)\in \mathbb{R}^n\setminus\mathcal{U}(t),\forall t\in[0,t_f]\big) \wedge\big(x(t_f)\in \mathcal{R}\big).
\end{align}
Thus, the true system trajectory satisfies the prescribed-time reach--avoid 
specification of Definition~\ref{def:ptrs}.

\section{Stability and Invariance Analysis}\label{sec:analysis}
This section establishes the stability of the closed-loop system and demonstrates set invariance with respect to the true obstacles.
We first verify the well-posedness of the CBF-QP \eqref{eq:cbf_qp_general} and continuity of the control input \(u\). 
Next, we show that the true system state remains within the \emph{VCZ} for all \(t \in [0,t_f]\). 
Finally, leveraging this confinement, we prove that the true system satisfies the PT-RA condition [Def. \ref{def:ptrs}] via set invariance.

\begin{lem}[Boundedness of solution of CBF-QP]
\label{lem:bounded_u}
Consider the virtual system \eqref{eqn:virtual_model} with locally Lipschitz \(f_c\) and \(g_c\).
The virtual control \(u_c\) derived from the CBF-QP \eqref{eq:cbf_qp_general} using the CBF vector \eqref{eqn:cbf_vector} satisfying \(\frac{\partial h}{\partial c}g_c(c)\neq 0\), admits a unique optimizer \(u_c^\ast(c,t)\) for every
\(c\in\mathcal R_s(0)\text{ and }t\in [0,\;t_f]\). Moreover, under Assumption~\ref{ass:non_empty_u},
\(\|u_c^\ast\|\) is uniformly bounded on \(\mathcal{R}_s(0)\) and piecewise continuous in \(t\).
\end{lem}

\begin{lem}[Continuity of controller]
\label{lem:continuty_of_u}
For a system \eqref{eqn:sys_dyn} satisfying Assumption~\ref{ass: assm1}, the control input \( u \in \mathbb{R}^n \) defined in equation \eqref{eqn:control} 
is continuous over error \(e\) \eqref{eqn:error_def}.
\end{lem}
The proofs of Lemma~\ref{lem:bounded_u} and Lemma~\ref{lem:continuty_of_u} are provided in Appendix~\ref{app:a} and Appendix~\ref{app:b}, respectively.
 
\begin{theorem}
\label{thm:invariance}
Consider an unknown dynamical system \eqref{eqn:sys_dyn} subject to Assumption~\ref{ass: assm1}, initialized at \(x(0)\in\mathbb{R}^n\setminus\mathcal{U}(0)\).
Let the control input \(u\) be defined in \eqref{eqn:control},
\[
u := 
\begin{cases} 
-k\zeta(\hat{e}) \phi(e) & \text{if } \|x-c\| > 0 \\
0_n & \text{if } \|x-c\| = 0,
\end{cases}
\]
where \(c\) is derived from the virtual dynamics \eqref{eqn:virtual_model}
\[
 \dot{c} = f_c(c) + g_c(c) u_c,
\]
of the \emph{VCZ} constructed in Section \ref{subsec:VCZ_cons}, such that \(x(0)\in \mathcal{C}(0)\subset\mathbb{R}^n\setminus\mathcal{U}(0)\), with \(f_c\) and \(g_c\) locally Lipschitz. The virtual control input \(u_c\) is derived from the CBF-QP \eqref{eq:cbf_qp_general} 
\[
\begin{aligned}
\min_{u_c \in \mathbb{R}^m} \quad & \tfrac{1}{2} u_c^\top H u_c + F^\top u_c \\
\text{s.t.} \quad & \dot{h}(c,t) + \Gamma(h(c,t)) \geq 0,
\end{aligned}
\]
where \(h(c,t) := [h_1(c,t), h_2(c,t), \dots, h_d(c,t)]^\top.\)
Under Assumption~\ref{ass:non_empty_u},

The resulting closed-loop system satisfies the PT-RA task
\[
\big(x(t)\in \mathbb{R}^n\setminus\mathcal{U}(t),\ \forall t\in[0,t_f]\big)\ \wedge\ \big(x(t_f)\in \mathcal{R}\big).
\]
\end{theorem}
\begin{pf}
    We first show that \(x(t)\in\mathcal{C}(t)\;\forall t\in[0,t_f]\) and then prove prescribed-time reach-avoid (PT-RA) for the VCZ \(\mathcal{C}(t)\).
    
    \textit{Step 1: Confinement of $x(t)$ in $\mathcal{C}(t)$:} 
    Using error \eqref{eqn:error_def} and the true dynamics \eqref{eqn:sys_dyn}, we get error dynamics as 
\begin{equation}
\label{eqn:error_dynamics}
    \dot{e}=f(x)+g(x)u+\omega-\dot{c}.
\end{equation}
By Assumption~\ref{ass: assm1}, \(f(x)\), \(g(x)\), and $\omega$ are bounded. The virtual system satisfies \(\dot{c}=f_c(c)+g_c(c)u_c\), where \(f_c(c)\) and \(g_c(c)\) are locally Lipschitz. 
The CBF~\eqref{eqn:avoid_cbf} subject to the feasibility Assumption~\ref{ass:non_empty_u} ensures
\(c(t)\in \mathcal{R}_s(t)\;\forall t\in [0,t_f]\), making \(f_c\) and \(g_c\) bounded over the compact domain \(\mathcal{R}_s(t)\subset\mathbb{R}^n\;\forall t\in [0,t_f]\). By virtue of Lemma~\ref{lem:bounded_u}, the virtual control input $u_c$ is uniformly bounded and piecewise continuous in \(t\).

Thus \(\dot{c}\) is uniformly 
bounded over compact set \(\mathcal{R}_s(0)\) and piecewise continuous in \(t\), implying continuity of \(c(t)\) and therefore of \(e(t)\).
With the control law~\eqref{eqn:control},
\begin{equation}
\label{eqn:error_dynamics_final}
    \dot{e}=m(e,c)-p(e,c)\,\big(k\,\zeta(\hat{e})\,\phi(e)\big),
\end{equation}
where \(m(e,c):=f(x)+\omega-\dot{c}\) is bounded and piecewise continuous in \(t\) \(\forall t \in [0,t_f]\) 
and \(p(e,c)=g(x)\) is sign-definite and bounded (Assumption~\ref{ass: assm1}).

At \(t=0\) the dynamics \eqref{eqn:error_dynamics_final} is well defined as control \(u\) is defined for \(x(0)\in \mathcal{C}(0)\) and remains well posed within  \(\mathcal{B}^\circ(0_n,r_c)\) as \(e\) is continuous over \(\mathcal{B}^\circ(0_n,r_c)\) from Lemma \ref{lem:continuty_of_u}.
Now to prove that \(x(t)\in \mathcal{C}(t)\;\forall t \in [0,t_f]\Longleftrightarrow e(t)\in \mathcal{B}^\circ(0_n,r_c),\forall t \in [0,t_f]\)
we consider the Lyapunov candidate \(V(e):=\zeta^2(\hat{e})\) defined over domain \(\mathcal{B}^\circ(0_n,r_c)\), where \(\hat e = \|e\|/r_c\). 
 \(\zeta(\cdot)=\ln\big(\frac{1+\hat{e}}{1-\hat{e}}\big)\) is an extended class-\(\mathcal{K}_\infty\) function (strictly increasing, continuous, and \(\zeta(0)=0\)), 
\(V(e)\) is positive in \(e\). Differentiating along \eqref{eqn:error_dynamics_final} and applying standard comparison arguments, 
we obtain ultimate boundedness of \(e(t)\), i.e., the closed-loop system is uniformly ultimately bounded [ref. \cite{khalil2002nonlinear} Theorem 4.18 ] on \(\mathcal{B}^\circ(0_n,r_c)\), which implies
\begin{equation}
   e(t)\in \mathcal{B}^\circ(0_n,r_c) ,\forall t \in [0,t_f]\Longleftrightarrow x(t)\in \mathcal{C}(t)\;\forall t \in [0,t_f]. 
   \label{eqn:confinement}
\end{equation}

The proof for error dynamics \eqref{eqn:error_dynamics_final} to be Uniformly Ultimately Bounded (UUB) within \(\mathcal{B}^\circ(0_n,r_c)\) follows similarly to related constructions in the literature (e.g., \cite{demos_funnel}).

\textit{Step 2: Ensuring PT-RA for VCZ $\mathcal{C}(t)$:} 
The CBF constraints for obstacle avoidance enforce forward invariance over the inflated set such that
\(c(t)\in \mathbb{R}^n\setminus \mathcal{U}_{VCZ}(t)\), which implies $
\|c(t)-b_{u_j}(t)\| \ge r_{u_j}+r_c,\quad j\in\{1,2,\dots,d-1\}, \forall t\in[0,t_f]$, using Lemma \ref{lem:cbfqp} and by Assumption \ref{ass:non_empty_u}.
From \eqref{eqn:confinement}, we have \(\|x(t)-c(t)\|< r_c,\forall t\in [0,t_f]\) (i.e., \(x(t)\in \mathcal{C}(t)\)).
By the triangle inequality, we have 
$
\|x(t)-b_{u_j}(t)\| \ge(r_{u_j}+r_c)-r_c  = r_{u_j},\;\forall t\in [0,t_f]$, implying
\begin{equation}
\label{eqn:cons1}
    x(t)\in \mathbb{R}^n\setminus \mathcal{U}(t)\quad \forall t\in [0,t_f].
\end{equation}

At \(t=t_f\), \(\mathcal{R}_s(t_f)=\mathcal{B}(b_R,r_r(t_f))\); forward invariance from shrinking barrier implies $\|c(t_f)-b_R\| \le r_r(t_f) \le r_R-r_c$, where $r_R \geq r_c$.
Again by the triangle inequality and \(x(t_f)\in \mathcal{C}(t_f)\) \eqref{eqn:confinement}, one has
        $\|x(t_f)-b_R\| \le\|x(t_f)-c(t_f)\|+\|c(t_f)-b_R\|
    \le r_c+(r_R-r_c)=r_R$,
which implies 
\begin{equation}
\label{eqn:cons3}
    x(t_f)\in \mathcal{R}.
\end{equation}
From \eqref{eqn:cons1} and \eqref{eqn:cons3}, we conclude
\[
\big(x(t)\in \mathbb{R}^n\setminus\mathcal{U}(t),\ \forall t\in[0,t_f]\big)\ \wedge\ \big(x(t_f)\in \mathcal{R}\big).
\]
\end{pf}

\section{Simulation Results and Discussion}
We validate the proposed control strategy on a two-dimensional second-order nonlinear control-affine system with unknown drift and a time-varying unsafe set. The true system dynamics are given by \(\dot{x} = f(x) + g(x)u + \omega\), where \(x(t) = [x_1(t),\, x_2(t)]^\top \in \mathbb{R}^2\), \(u(t) = [u_1(t),\, u_2(t)]^\top \in \mathbb{R}^2\), \(f(x) = [5\sin(x_1 x_2),\, 5\cos(x_1 x_2)]^\top\), \(\omega = [0.4\cos(t),\, 0.4\sin(t)]^\top\), and \(g = \text{diag}(0.8,\, 0.5)\) with initial condition \(x(0) = [0,\, 0]^\top\). A static and a dynamic circular obstacle are defined as the sets \(\mathcal{U}_1 = \mathcal{B}^\circ(b_{u_1}, r_{u_1})\) and \(\mathcal{U}_2 = \mathcal{B}^\circ(b_{u_2}, r_{u_2})\), where \(b_{u_1} = [1.5,\, 2]^\top\), \(b_{u_2}(t) = [5 + 0.4t,\, 5 - 0.4t]^\top\), \(r_{u_1} = 0.5\), and \(r_{u_2} = 1.5\). The objective is to reach the target set \(\mathcal{R} = \mathcal{B}(b_R, 1.1)\) in \(t_f = 10~\text{s}\), where \(b_R = [10,\, 10]^\top\).

The virtual trajectory \(c(t) = [c_1(t), c_2(t)]^\top\) evolves under the single-integrator dynamics \(\dot{c}(t) = u_c(t)\) given in~\eqref{eqn:virtual_model}, where the virtual control \(u_c(t)\in\mathbb{R}^2\) is obtained from the CBF-QP~\eqref{eq:cbf_qp_general}. The safety constraints for the virtual state are encoded by the CBFs \(h_1(c,t) := \|c-b_{u_1}\|^2-(r_c+r_{u_1})^2\) and \(h_2(c,t) := \|c-b_{u_2}(t)\|^2-(r_c+r_{u_2})^2\) with \(r_c = 0.5\), in accordance with~\eqref{eqn:avoid_cbf}. To impose prescribed-time convergence, the target set is chosen as \(\mathcal{R}_{VCZ}=\mathcal{B}(b_R,0.5)\subset\mathcal{R}\) with center \(b_R=[10,10]^\top\), together with a shrinking set \(\mathcal{R}_s(t)=\mathcal{B}(b_R,r_r(t))\) and the reach CBF \(h_3(c,t) := r_r^2(t)-\|c-b_R\|^2\) defined in~\eqref{eqn:reach_cbf}. The radius \(r_r(t)\) decreases linearly as \(r_r(t)=\big(r_{r_{t_f}}-r_{r_{t_0}}\big)\frac{t}{t_f}+r_{r_{t_0}}\) from \(r_{r_{t_0}}=15\) to \(r_{r_{t_f}}=0.5\) over the interval \([0,t_f]\) with \(t_f=10\,\text{s}\). Given an initial condition \(x(0)\), the virtual state is initialized as \(c(0)=x(0)\) with \(\|x(0)-c(0)\|<r_c\), ensuring that the real system starts inside the VCZ. The CBF-QP ensures that \(c(t)\) remains within the safe sets and inside the shrinking region \(\mathcal{R}_s(t)\), while the actual control \(u(t)\) confines the true state \(x(t)\) around \(c(t)\). Consequently, when \(c(t)\) reaches \(\mathcal{R}_{VCZ}\) at \(t=t_f\), the true state also lies inside the desired target region \(\mathcal{R}\), thereby achieving prescribed-time reach-avoid.

\begin{figure}[ht]
    \centering
   {\includegraphics[width=0.49\columnwidth]{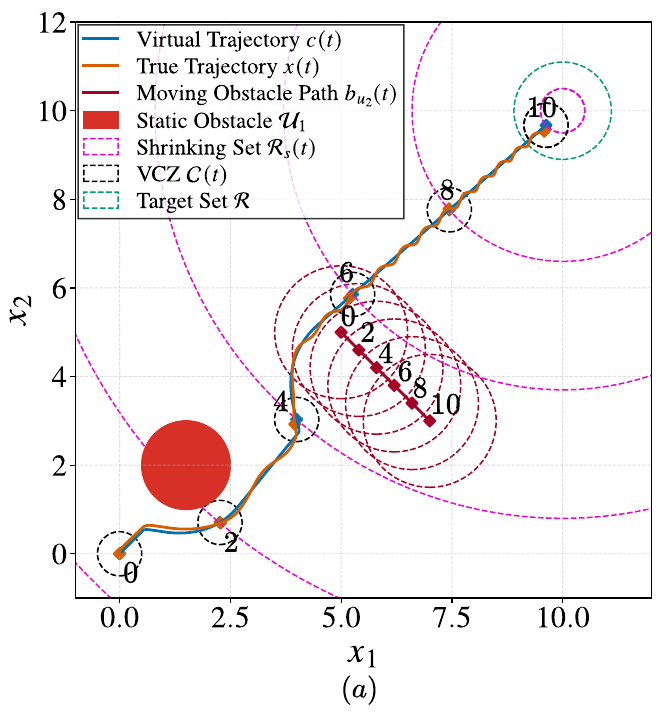}\label{fig:trajectory}}
    \hfill
  {\includegraphics[width=0.49\columnwidth]{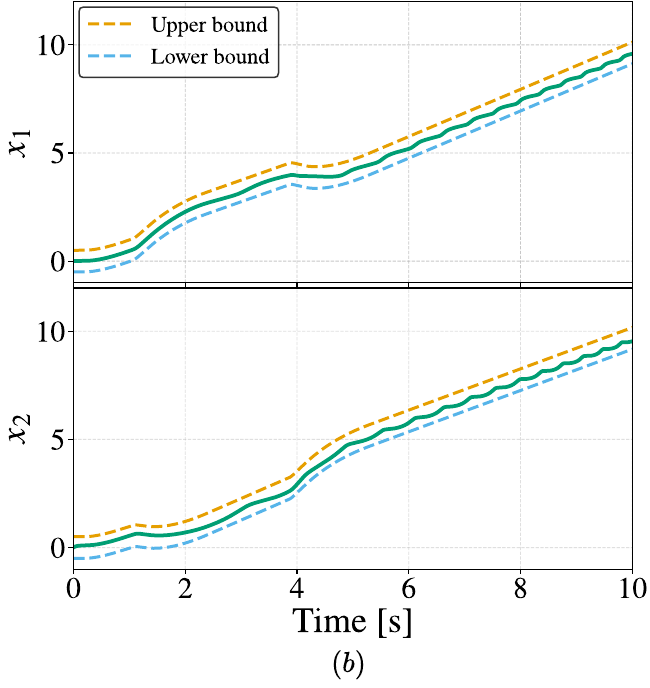}\label{fig:trajectory_bound}}
    \caption{(a) Trajectory plot in state space for PT-RA over \(10\)~s. 
Numbers next to each \colorsquare{mycolor} \colorsquare{mycolor1} \colorsquare{mycolor2} marker indicate time. (b) Time evolution of state trajectories within the virtual bounds ($c(t)\pm r_c)$.}
    \label{fig:ab}
\end{figure}

As illustrated in Fig.~\ref{fig:ab}(a) and (b), the system trajectory remains within the prescribed bounds over the entire time horizon \(t \in [0, 10]\). It consistently avoids both static and moving obstacles with a guaranteed safety margin of \(r_c = 0.5\), enforced by the control barrier function-based constraints. Moreover, the trajectory converges to the final target set \(\mathcal{R}\), demonstrating the effectiveness of the proposed approach under time-varying safety and performance constraints for an unknown system. The simulation video is available at https://youtu.be/UXx8Hr4aBio .

\section{Discussion and Comparison}

Table~\ref{tab:sttcbf_comparison} compares the proposed VCZ–CBF framework with representative safe control approaches. Classical HJ–reachability methods~\citep{waberisich2023,mitchell2005} provide formal safety guarantees but rely on accurate models, known disturbance bounds, and are, in general, computationally demanding and less suited for real-time control in dynamic environments. Predictive safety filters (PSFs) ~\citep{waberisich2023} are real-time but require nominal control and cannot work under a dynamic unsafe set unless the trajectory is known. Reinforcement learning approaches~\citep{saxena} can handle partially unknown systems but depend on extensive pretraining or precomputation and no formal guarantee. Similarly, Symbolic Control~\citep{tabuada2009verification} also requires extensive precomputation and is difficult to achieve PT-RA. Funnel control~\citep{4639441} ensures bounded transient response without explicit model knowledge, but cannot directly encode time-varying safety constraints or moving obstacles. Spatio-Temporal Tubes (STT)~\citep{dasstt} ensures PT-RA with an unknown model, but also cannot handle dynamic environments. {The model-free framework of \cite{molnar2021} leverages reduced-order kinematic models that naturally accommodate underactuated systems, whereas the proposed VCZ-CBF requires square, sign-definite \(g(x)\)(Assumption~\ref{ass: assm1}). However, the model-free variant in \cite{molnar2021} provides only ISSf guarantees, where the safety margin depends on uncompensated model-dependent terms that are generally unknown without the dynamical model.}

In contrast, the proposed VCZ--CBF framework is approxi-mation-free, ensures PT-RA, operating on a nominal virtual system while confining the true, unknown dynamics within a shrinking invariant region. This decoupling enables real-time implementation, prescribed-time reach-avoid guarantees, and robustness in dynamic environments, capabilities not jointly achieved by existing methods.

Despite these advantages, several practical limitations remain. First, as in many QP-based CBF methods, the virtual control $u_c$ can become large near intersecting constraints or rapidly moving obstacle boundaries, increasing control effort. Second, estimating obstacle velocities in real time is challenging under noisy or delayed perception, which may lead to conservative barrier constraints and potential control saturation. Finally, because the true state must remain within a fixed radius $r_c$ of the virtual trajectory, the tightened safe sets reduce the effective workspace, introducing spatial conservatism. Incorporating bounded-input feasibility and adaptive inflation strategies could mitigate these issues while preserving invariance and prescribed-time guarantees.

\begin{table}[t]
\centering
\caption{Comparison of proposed approach with classical algorithms}
\label{tab:sttcbf_comparison}
\renewcommand{\arraystretch}{1.05} 

\resizebox{\columnwidth}{!}{%
\begin{threeparttable}
\begin{tabular}{lcccccc}
\toprule
\textbf{Algorithm} & \textbf{Unknown} & \textbf{Precomp./} & \textbf{Formal} & \textbf{PT-RA} & \textbf{Dynamic}\\
                   & \textbf{Dynamics} & \textbf{Training Req.} & \textbf{Guarantee} & \textbf{} & \textbf{Environments} \\
\midrule

RL \citep{saxena}                & \cmark & \cmark & \xmark & \xmark & \xmark  \\
Symbolic Control \citep{tabuada2009verification}  & \xmark & \cmark & \cmark & \xmark & \xmark\\
PSFs \citep{waberisich2023}               & \xmark & \xmark & \cmark & \xmark & \xmark  \\
STT \citep{dasstt}               & \cmark & \xmark & \cmark & \cmark & \xmark  \\

HJ-Reachability \citep{mitchell2005}   & \xmark & \cmark & \cmark & \xmark & \xmark  \\
Funnel Control \citep{Berger2021}   & \cmark & \xmark & \cmark & \xmark & \xmark \\
{Model-free Safety} \citep{molnar2021} & \cmark$^1$ & \xmark & \cmark$^1$ & \xmark & \xmark \\
\textbf{VCZ-CBF (Proposed)} & \cmark & \xmark & \cmark & \cmark & \cmark  \\
\bottomrule
\end{tabular}
\begin{tablenotes}\footnotesize
\item[1] {Does not use the full dynamical model in the control law, 
but selecting gain requires knowledge of  the bounds of dynamic parameters
and the model-free variant guarantees only ISSf.}
\end{tablenotes}

\end{threeparttable}
}
\end{table}

\section{Conclusion}
We proposed a safe control framework for nonlinear control-affine systems with unknown dynamics using a \emph{Virtual Confinement Zone} (VCZ) scheme. A virtual control input was synthesized via a CBF-based Quadratic Program, ensuring prescribed-time reach–avoid guarantees without requiring explicit system identification or uncertainty bounds. Safety was enforced through Control Barrier Functions applied to a known virtual system, while the true system was confined around the virtual trajectory.

\bibliography{references_2}            

\appendix

\section{Proof of Lemma \ref{lem:bounded_u}}    %
\label{app:a}
\begin{pf}
    The obstacle avoidance CBF condition from \eqref{eqn:avoid_cbf} can be written as:
    \[
    a_j(c,t)^\top u_c \;\geq\; \rho_j(c,t), \qquad j=1,2,\dots,d-1,
    \]
    with $a_j(c,t)=2\,g_c(c)^\top\!\big(c-b_{u_j}(t)\big)$, $
    \rho_j(c,t):=-\gamma\!\big(h_j(c,t)\big)-2\big(c-b_{u_j}(t)\big)^\top f_c(c)+2\big(c-b_{u_j}(t)\big)^\top \dot b_{u_j}(t)$.
    
    From Assumption~\ref{ass:qp_feasibility}, the signals $\dot b_{u_j}(t)$, $\|c-b_{u_j}(t)\|$, $\|f_c(c)\|$, and $\|g_c(c)\|$ are bounded (and measurable) on the compact domain $\mathcal{R}_s(0)$ enforced by shrinking CBF \eqref{eqn:avoid_cbf}; hence $\|\rho_j(\cdot)\|$ is bounded for $j=1,\dots,d-1$.

    The CBF for target set reach \eqref{eqn:reach_cbf} can be expressed as 
    \[
    a_j(c,t)^\top u_c \;\geq\; \rho_j(c,t), \qquad j=d,
    \]
    with
    $a_j(c,t)=2\,g_c(c)^\top\!\big(b_R-c\big)$, $
    \rho_j(c,t):=-\gamma\!\big(h_j(c,t)\big)-2\big(b_R-c\big)^\top f_c(c)\;-\;2\,r_r(t)\,\dot r_r(t)$.
    By construction of the shrinking set $\mathcal{R}_s(t)$, the quantities $\|c-b_R\|$, $|r_r(t)|$, $|\dot r_r(t)|$, $\|f_c(c)\|$, and $\|g_c(c)\|$ are bounded (and measurable) on $\mathcal{R}_s(0)$; hence $\|\rho_j(\cdot)\|$ is bounded for $j=d$.

    Stacking them together, the constraint becomes
    \[
       A(c,t)\,u_c \;\geq\; \rho(c,t),
    \]
    where the rows of $A$ are the $a_j^\top$ and the entries of $\rho$ are the $\rho_j$ with \(\|\rho\|\) bounded. 

    From the feasibility Assumption~\ref{ass:non_empty_u} and strict convexity of the objective (constant $H\succ0$), bounded RHS \(\|\rho\|\) and non vanishing LHS \( A(c,t)\,u_c \neq0\) as \(\frac{\partial h}{\partial c}g_c\neq0\), the CBF–QP \eqref{eq:cbf_qp_general} admits a unique optimizer $u_c^\ast(c,t)$. Since $H\succ0$ the quadratic cost is coercive, and the affine feasible set $\{u_c: A(c,t)\ge \rho(c,t)\}$ is a nonempty closed polyhedron; by Berge’s maximum theorem and assumption \ref{ass:non_empty_u} the argmin set is nonempty and compact, hence $\|u_c^\ast\|$ is uniformly bounded on $\mathcal{R}_s(0)$. Finally, by assumption \ref{ass:non_empty_u} the solution is feasible \(\forall t\in [0,t_f]\) and the optimizer is unique, the mapping $t \mapsto u_c^\ast(c,t)$ is piecewise continuous in $t$.
\end{pf}
\section{Proof of Lemma \ref{lem:continuty_of_u}}  
\label{app:b}
\begin{pf} 
From \eqref{eqn:control} the auxiliary function \( \zeta(\hat{e}) \) is continuous on \( \mathcal{B}^\circ(0_n,r_c) \) 
with respect to \(e\). Furthermore, the function \( \phi(e)=\frac{e}{\|e\|} \) 
is continuous on \( \mathbb{R}^n \setminus \{0_n\} \), over \(e\) and \(k\) is constant based on {known sign of eigenvalues of} \(\frac{g(x)+g(x)^\top}{2}\).

We prove the continuity of the control input \(u\) in two parts: one  where \(\|e\|\neq0\), and another at \(\|e\|=0\) by applying limits.

\emph{Part 1.}  
For \(\|e\|\neq0\) , \(\zeta(\hat{e}) \) and \( \phi(e) \) are continuous in \( \mathcal{B}^\circ(0_n,r_c) \setminus \{0_n\} \) in \(e\) , the controller \(u\) is continuous in \( \mathcal{B}^\circ(0_n,r_c) \setminus \{0_n\} \) in \(e\).

\emph{Part 2.}  
To prove that \( u \) is continuous at \(\|e\| = 0\), it is necessary that
\(
u_{\|e(t)\|=0} = \lim_{\|e(t)\|\to 0} u.\)

When \(\|e(t)\|=0\), the control input is defined as \(u=0\). Thus
\(u_{\|e(t)\|=0}=0\). For \(\|e\|>0\), from \eqref{eqn:control},
\[
u = -k\,\zeta(\hat{e})\,\phi(e)
= -k\left(\ln\!\left(\tfrac{1+\hat{e}}{1-\hat{e}}\right)\right)\cdot\frac{e}{\|e\|}.
\]
For \(|\hat{e}|\leq 1\), expanding \(\zeta\) gives $\zeta \approx 2\hat{e} + \tfrac{2}{3}\hat{e}^3 + \dots$
so
\[
u \approx -2k\left(\tfrac{\|e\|}{r_c} + \tfrac{\|e\|^3}{3r_c^3}+\dots\right)\cdot\frac{e}{\|e\|}.
\]
The terms \(\|e\|\) cancel, yielding \(u\to 0_n\) as \(\|e\|\to 0\).

\medskip
Therefore, \(\lim_{\|e\|\to 0}u=0_n=u_{\|e\|=0}\). 
Hence, the control input \(u\) is continuous in \(e\) in \( \mathcal{B}^\circ(0_n,r_c) \).

\end{pf}

\end{document}